\RequirePackage{fix-cm}

\documentclass[smallextended]{svjour3}       

\smartqed  

\usepackage{amsfonts,amssymb,amsmath}
\usepackage{graphics,graphicx,epsfig}
\usepackage[dvipsnames]{xcolor}
\usepackage[normalem]{ulem}
\usepackage{comment}
\usepackage{epstopdf}
\usepackage{subfigure}
\usepackage{marvosym}
\usepackage{hyperref}
\usepackage{geometry}
\usepackage{cite}
\usepackage [latin1]{inputenc}

\def\be{\begin{eqnarray}}
\def\ee{\end{eqnarray}}
\def\ba{\begin{array}{l}}
\def\ea{\end{array}}

\begin{document}
\title{Separability criteria based on Bloch representation of density matrices}
\author{Hui Zhao$^1$ \and Mei-Ming Zhang $^1$  \and Naihuan Jing$^{2,3}$  \and Zhi-Xi Wang$^4$ }
\institute{ \Letter~~Hui Zhao \at
                     zhaohui@bjut.edu.cn \and
              \Letter~~Mei-Ming Zhang \at
                     603171581@qq.com \and
            \Letter~~    NaiHuan Jing\at
                    ~~~ nathanjing@hotmail.com\and
              \Letter~~     Zhi-Xi Wang\at
                     ~~~wangzhx@cnu.edu.cn
           \and
 \at 1 College of Applied Sciences, Beijing University of Technology, Beijing 100124, China\\
 \at 2 Department of Mathematics, North Carolina State University, Raleigh, NC 27695, USA\\
 \at 3 Department of Mathematics, Shanghai University, Shanghai 200444, China\\
 \at 4 School of Mathematical Sciences, Capital Normal University, Beijing 100048, China
}


\maketitle
\begin{abstract}
\baselineskip18pt
We study separability criteria in multipartite quantum systems of arbitrary dimensions by using the Bloch representation of density matrices. We first derive the norms of the correlation tensors and obtain the necessary conditions for separability under partition of tripartite and four-partite quantum states. Moreover, based on the norms of the correlation tensors,
we obtain the separability criteria by matrix method. Using detailed examples, our results are seen to be able to detect more entangled states than previous studies. Finally, necessary conditions of separability for multipartite systems are given under arbitrary partition.
\end{abstract}

\section{Introduction}
Quantum entanglement is a significant feature in quantum physics. It plays an important role in many ways such as quantum teleportation \cite{ea} and cryptography \cite{hbb, grz, shh}. Therefore, determination of a state being entangled or not is a significant issue in quantum information theory. Although many efforts have been devoted to the study of this problem \cite{wcz, cw, gl, hr, hj, fwz, xzw, zzf, zff, yz, zgj}, it is still open except some special cases. Since norms of the Bloch vectors have a close relationship to separability criteria, research on related problem has attracted attention recently.

The norms of correlation tensors in the Bloch representation of quantum states were used to improve the separability criterion \cite{bf, bk}. The norms of correlation tensors for density matrices in lower dimensions were discussed in \cite{js, kg}. In \cite{lww}, the norms of the correlation tensors for quantum state with systems less than or equal to four have been investigated. Some sufficient or necessary conditions of separability by using the norms of the correlation tensors of density matrices were presented in \cite{jid, jiv, asm, lwl}.  The relations between the norms of the correlation tensors and the detection of genuinely multipartite entanglement in tripartite quantum systems have also been established in \cite{ljw}. In \cite{ys}, Sufficient and necessary condition of full separability for 3-qubit systems was derived. The relation among bipartite concurrence, concurrence of assistance and genuine tripartite entanglement for $2\otimes2\otimes n$ dimensional quantum states was presented in \cite{yss}.
In \cite{vh}, non full separability criterion in multipartite quantum systems based on correlation tensors was discussed. Using correlation tensors, the authors in \cite{syl} have provided full separability criteria for bipartite and multipartite quantum states.

In this paper, we study necessary conditions of separability for multipartite quantum states based on correlation tensors. It is known that any n-partite pure state that can be written as a tensor product $|\varphi\rangle\langle\varphi|=|\varphi_A\rangle\langle\varphi_A|\otimes |\varphi_{\overline{A}}\rangle\langle\varphi_{\overline{A}}|$ with respect to some bipartition $A\overline{A}$ ($A$ denoting some subset of subsystems and $\overline{A}$ its complement) is called biseparable. And any mixed state that can be decomposed into a convex sum of biseparable pure states is called biseparable. Consequently, any non-biseparable mixed state is called genuinely multipartite entangled. However, we mainly study separability under specific partition rather than biseparable case for multipartite quantum states. Therefore, our methods can detect entangled states rather than genuinely multipartite entangled states. In Section 2, we present an upper bound for the norm of correlation tensors and separability criteria under any partition by constructing matrices for tripartite quantum states. By a detailed example, our results are seen to outperform previously published results. In Section 3, we generalize an inequality of the norm of the correlation tensors for four-partite states and derive the necessary conditions of separability under different partition for four-partite quantum states. We also give examples to show that our criteria can detect more entangled states than previous available results. In Section 4, we generalize the norm of correlation tensors to multipartite quantum systems and obtain necessary conditions of separability under k-partitions. Comments and conclusions are given in Section 5.

\section{Separability criteria for tripartite Quantum States}\label{sec2}
We first consider the separability criteria for tripartite quantum states. Let $H_n^{d_n}$ ($n=1,2,3$) be $d_n$-dimensional Hilbert spaces. Let $\lambda_{i_f}^{(f)}$, $i_f=1,\cdots,d_f^2-1$, $f=1,2,3,$ denote the mutually orthogonal generators of the special unitary Lie algebra $\mathfrak{su}{(d_f)}$ under
a fixed bilinear form, and $I$ the $d_m\times d_m$ identity matrix ($m=1,2,3$). A tripartite state $\rho\in H_1^{d_1}\otimes H_2^{d_2}\otimes H_3^{d_3}$ can be written as follows:
\begin{eqnarray}
&&\rho=\frac{1}{d_1d_2d_3}I\otimes I\otimes I+\sum_{f=1}^3\frac{d_f}{2d_1d_2d_3}\sum_{i_f=1}^{d_f^2-1}t_{i_f}^f\lambda_{i_f}^{(f)}\otimes I\otimes I+\\
&&+\sum_{1\leq f<g\leq 3}\frac{d_fd_g}{4d_1d_2d_3}\sum_{i_f=1}^{d_f^2-1}\sum_{i_g=1}^{d_g^2-1}t_{i_fi_g}^{fg}\lambda_{i_f}^{(f)}\otimes\lambda_{i_g}^{(g)}\otimes I+\frac{1}{8}\sum_{i_1=1}^{d_1^2-1}\sum_{i_2=1}^{d_2^2-1}\sum_{i_3=1}^{d_3^2-1}t_{i_1i_2i_3}^{123}\lambda_{i_1}^{(1)}\otimes\lambda_{i_2}^{(2)}\otimes\lambda_{i_3}^{(3)},\nonumber
\end{eqnarray}
where $\lambda_{i_f}^{(f)}$ or $\lambda_{i_g}^{(g)}$ ($(f)$ or $(g)$ refers the position of $\lambda_{i_f}$ or $\lambda_{i_g}$ in the tensor product) stands for the operators with $\lambda_{i_f}$ or $\lambda_{i_g}$ on $H_{d_f}$ or $H_{d_g}$, and $I$ on the respective spaces, $t_{i_f}^f=tr(\rho\lambda_{i_f}^{(f)}\otimes I\otimes I),$ $t_{i_fi_g}^{fg}=tr(\rho\lambda_{i_f}^{(f)}\otimes\lambda_{i_g}^{(g)}\otimes I),$ $t_{i_1i_2i_3}^{123}=tr(\rho\lambda_{i_1}^{(1)}\otimes\lambda_{i_2}^{(2)}\otimes\lambda_{i_3}^{(3)})$. Let $T^{(f)}, T^{(fg)}, T^{(123)}$ be the vectors (tensors) with entries $t_{i_f}^f, t_{i_fi_g}^{fg}, t_{i_1i_2i_3}^{123}$ respectively. And $\|\cdot\|$ stand for the Hilbert-Schmidt norm or Frobenius norm, then we have $\|T^{(f)}\|^2=\sum_{i_f=1}^{d_f^2-1}(t_{i_f}^f)^2$, $\|T^{(fg)}\|^2=\sum_{i_f=1}^{d_f^2-1}\sum_{i_g=1}^{d_g^2-1}(t_{i_fi_g}^{fg})^2$ and $\|T^{(123)}\|^2=\sum_{i_1=1}^{d_1^2-1}\sum_{i_2=1}^{d_2^2-1}\sum_{i_3=1}^{d_3^2-1}(t_{i_1i_2i_3}^{123})^2$. The trace norm is defined as the sum of the singular values of the matrix $A\in\mathbb{R}^{m\times n}$, i.e., $\|A\|_{tr}=\sum_i\sigma_i=tr\sqrt{A^\dag A}$, where $\sigma_i$, $i=1,\cdots,min(m,n)$, are the singular values of the matrix $A$ arranged in descending order.

In particular, for any pure state $\rho\in H_1^{d_1}\otimes H_2^{d_2}$, $2\leq d_1\leq d_2$, we have $\rho=\frac{1}{d_1d_2}I_{d_1}\otimes I_{d_2}+\frac{1}{2d_2}\sum_{i_1=1}^{d_1^2-1}t^1_{i_1}\lambda_{i_1}^{(1)}\otimes I_{d_2}+\frac{1}{2d_1}\sum_{i_2=1}^{d_2^2-1}t^2_{i_2}I_{d_1}\otimes \lambda_{i_2}^{(2)}+\frac{1}{4}\sum_{i_1=1}^{d_1^2-1}\sum_{i_2=1}^{d_2^2-1}t^{12}_{i_1i_2}\lambda_{i_1}^{(1)}\otimes \lambda_{i_2}^{(2)}$.
\begin{lemma}\label{lemma:1}
Let $\rho\in H_1^{d_1}\otimes H_2^{d_2}$ be a pure state, for $d_1\leq d_2$,
\begin{eqnarray}
\|T^{(12)}\|^2\leq \frac{4(d_2^2-1)}{d_2^2}.
\end{eqnarray}
\end{lemma}
\begin{proof}
Let $\rho_1$ and $\rho_2$ be the density matrices with respect to the subsystem $H_1$ and $H_2$. For a pure state $\rho$, we have $tr(\rho^2)=1$ and $tr(\rho_1^2)=tr(\rho_2^2)$, i.e.,
\be
&&tr(\rho^2)=\frac{1}{d_1d_2}+\frac{1}{2d_2}\|T^{(1)}\|^2+\frac{1}{2d_1}\|T^{(2)}\|^2+\frac{1}{4}\|T^{(12)}\|^2=1,\nonumber\\
&&tr(\rho_1^2)=\frac{1}{d_1}+\frac{1}{2}\|T^{(1)}\|^2=tr(\rho_2^2)=\frac{1}{d_2}+\frac{1}{2}\|T^{(2)}\|^2.
\ee
Therefore
\begin{eqnarray}
\|T^{(12)}\|^2=\frac{4(d_2^2-1)}{d_2^2}-\frac{2(d_1+d_2)}{d_1d_2}\|T^{(2)}\|^2\leq\frac{4(d_2^2-1)}{d_2^2}.
\end{eqnarray}
\begin{flushright}
$\square$
\end{flushright}
\end{proof}
\begin{proposition}\label{prop:1}
Let $\rho\in H_1^{d_1}\otimes H_2^{d_2}\otimes H_3^{d_3}$ be a pure state, for $d_1\leq d_2\leq d_3$,
\begin{eqnarray}
\|T^{(123)}\|^2\leq8(1-\frac{d_1d_2+d_1d_3+d_2d_3-d_1-d_2}{d_1d_2d_3^2}).
\end{eqnarray}
\end{proposition}
\begin{proof}
Let $\rho_{i_1}$ and $\rho_{i_2i_3}$ be the density matrices with respect to the subsystem $H_{d_{i_1}}$, $i_1=1, 2, 3$, and $H_{d_{i_2}d_{i_3}}$, $1\leq i_2<i_3\leq3$. For a pure state $\rho$, we have $tr(\rho^2)=1$ and $tr(\rho_{i_1}^2)=tr(\rho_{i_2i_3}^2)$, i.e.,
\begin{eqnarray}
tr(\rho^2)&=&\frac{1}{d_1d_2d_3}+\frac{1}{2}(\frac{1}{d_2d_3}\|T^{(1)}\|^2+\frac{1}{d_1d_3}\|T^{(2)}\|^2+\frac{1}{d_1d_2}\|T^{(3)}\|^2)+\frac{1}{4}(\frac{1}{d_3}\|T^{(12)}\|^2\nonumber\\
&+&\frac{1}{d_2}\|T^{(13)}\|^2+\frac{1}{d_1}\|T^{(23)}\|^2)+\frac{1}{8}\|T^{(123)}\|^2=1,\\
tr(\rho_{i_1}^2)&=&\frac{1}{d_{i_1}}+\frac{1}{2}\|T^{(i_1)}\|^2=tr(\rho_{i_2i_3}^2)=\frac{1}{d_{i_2}d_{i_3}}+\frac{1}{2}(\frac{1}{d_{i_3}}\|T^{(i_2)}\|^2+\frac{1}{d_{i_2}}\|T^{(i_3)}\|^2)+\frac{1}{4}\|T^{(i_2i_3)}\|^2\nonumber.
\end{eqnarray}
Therefore
\begin{eqnarray}
\|T^{(123)}\|^2&=&8(1-\frac{1}{d_1d_2d_3})-4(\frac{1}{d_2d_3}\|T^{(1)}\|^2+\frac{1}{d_1d_3}\|T^{(2)}\|^2+\frac{1}{d_1d_2}\|T^{(3)}\|^2)\nonumber\\
&-&2(\frac{1}{d_3}\|T^{(12)}\|^2+\frac{1}{d_2}\|T^{(13)}\|^2+\frac{1}{d_1}\|T^{(23)}\|^2)\nonumber\\
&\leq&8(1-\frac{d_1d_2+d_1d_3+d_2d_3-d_1-d_2}{d_1d_2d_3^2})-4[\frac{d_2(d_3-1)}{d_2d_3}\|T^{(1)}\|^2\nonumber\\
&+&\frac{d_1(d_3-1)}{d_1d_3}\|T^{(2)}\|^2+\frac{d_3+d_1d_2-(d_1+d_2)}{d_1d_3}\|T^{(3)}\|^2]\nonumber\\
&\leq&8(1-\frac{d_1d_2+d_1d_3+d_2d_3-d_1-d_2}{d_1d_2d_3^2}).
\end{eqnarray}
\begin{flushright}
$\square$
\end{flushright}
\end{proof}

{\bf Remark 1.} When $d_1=d_2=d_3=d$, we can obtain that $\|T^{(123)}\|\leq\frac{1}{d^3}(8d^3-24d+16)$. Proposition \ref{prop:1} is a generalization of the Theorem 1 given in \cite{lww}.

For the tripartite quantum state $\rho\in H_1^{d_1}\otimes H_2^{d_2}\otimes H_3^{d_3}$, we denote the bipartitions as follows: $f|gh, fg|h$ and $f\neq g\neq h\in\{1, 2, 3\}$. We first consider the separability of $\rho$ under bipartition $f|gh$.
Then we construct the matrix $S(\rho_{f|gh})$ by
\begin{eqnarray}
S(\rho_{f|gh})=
\left( \begin{array}{cccc}
            1& \displaystyle(T^{(g)})^t&\displaystyle(T^{(h)})^t& \displaystyle(T^{(gh)})^t \\
          \displaystyle T^{(f)} &\displaystyle T^{(fg)}&\displaystyle T^{(fh)}&\displaystyle T^{(fgh)}\\
           \end{array}
      \right ).
\end{eqnarray}
Using this matrix and the inequality for 1-body correlation tensors $\|T^{(j)}\|^2\leq \frac{2(d_j-1)}{d_j}(j=f, g, h)$\cite{vh} and Lemma \ref{lemma:1}, we get the following separability criterion.  
\begin{theorem}\label{thm:1}
If the state $\rho\in H_1^{d_1}\otimes H_2^{d_2}\otimes H_3^{d_3}$ is separable under the bipartition $f|gh$, then
\begin{eqnarray}
\|S(\rho_{f|gh})\|_{tr}\leq\sqrt{\frac{(3d_f-2)(9d_gd_h^2-2d_h^2-2d_gd_h-4d_g)}{d_fd_gd_h^2}}.
\end{eqnarray}
\end{theorem}
\begin{proof}
A tripartite mixed state $\rho$ is separable under bipartition $f|gh$ whenever it can be expressed as
\begin{eqnarray}
\rho_{f|gh}=\sum_lp_l\rho_l^{(f)}\otimes\rho_l^{(gh)},
\end{eqnarray}
where the probabilities $p_l>0, \sum_lp_l=1$. Let $\rho_l^{(f)}=\frac{1}{d_f}I_{d_f}+\frac{1}{2}\sum_{i_f=1}^{d_f^2-1}t_{li_f}^f\lambda_{i_f}^{(f)}$, $\rho_l^{(gh)}=\frac{1}{d_gd_h}I_{d_g}\otimes I_{d_h}+\frac{1}{2d_h}\sum_{i_g=1}^{d_g^2-1}t_{li_g}^g \lambda_{i_g}^{(g)}\otimes I_{d_h}+\frac{1}{2d_g}\sum_{i_h=1}^{d_h^2-1}t_{li_h}^hI_{d_g}\otimes \lambda_{i_h}^{(h)}+\frac{1}{4}\sum_{i_g=1}^{d_g^2-1}\sum_{i_h=1}^{d_h^2-1}t_{li_gi_h}^{gh}\lambda_{i_g}^{(g)}\otimes\lambda_{i_h}^{(h)}$. Let $v_l^f, v_l^g, v_l^h$ and $v_l^{gh}$ be the column vectors with entries $t_{li_f}^f, t_{li_g}^{g}, t_{li_h}^{h}$ and $t_{li_gi_h}^{gh}$ respectively. Therefore,
\begin{eqnarray}
&&T^{(f)}=\sum_lp_lv_l^f,\  T^{(g)}=\sum_lp_lv_l^g,\  T^{(h)}=\sum_lp_lv_l^h,\  T^{(fg)}=\sum_lp_lv_l^f(v_l^g)^t,\nonumber\\ &&T^{(fh)}=\sum_lp_lv_l^f(v_l^h)^t,\
T^{(gh)}=\sum_lp_lv_l^{gh},\  T^{(fgh)}=\sum_lp_lv_l^f(v_l^{gh})^t,
\end{eqnarray}
where $t$ stands for transpose.
Then the matrix $S(\rho_{f|gh})$ can be written as
\begin{eqnarray}
S(\rho_{f|gh})&=&\sum_lp_l
\left( \begin{array}{ccccccccccccccc}
            1& \displaystyle(v_l^g)^t& \displaystyle(v_l^h)^t& \displaystyle(v_l^{gh})^t \\
           \displaystyle v_l^f &\displaystyle v_l^f(v_l^g)^t&\displaystyle v_l^f(v_l^h)^t&\displaystyle v_l^f(v_l^{gh})^t\\
           \end{array}
      \right )
\nonumber\\
&=&\sum_lp_l
\left( \begin{array}{ccccccccccccccc}
              1\\
              \displaystyle v_l^f\\
            \end{array}
     \right )
\left( \begin{array}{ccccccccccccccc}
              1&\displaystyle(v_l^g)^t&\displaystyle(v_l^h)^t&\displaystyle(v_l^{gh})^t\\
            \end{array}
     \right ).
\end{eqnarray}
Thus,
\begin{eqnarray}
\|S(\rho_{f|gh})\|_{tr}&\leq&\sum_lp_l\|
\left( \begin{array}{ccccccccccccccc}
              1\\
              \displaystyle v_l^f\\
            \end{array}
     \right )
\left( \begin{array}{ccccccccccccccc}
              1&\displaystyle(v_l^g)^t&\displaystyle(v_l^h)^t&\displaystyle(v_l^{gh})^t\\
            \end{array}
     \right )\|_{tr}\nonumber\\
&=&\sum_lp_l\|
\left( \begin{array}{ccccccccccccccc}
                 1\\
               \displaystyle v_l^f\\
            \end{array}
     \right )\|
\|\left( \begin{array}{ccccccccccccccc}
                 1&\displaystyle(v_l^g)^t&\displaystyle(v_l^h)^t&\displaystyle(v_l^{gh})^t\\
            \end{array}
     \right )^t\|\nonumber\\
&=&\sum_lp_l\sqrt{1+\|v_l^f\|^2}\sqrt{1+\|v_l^g\|^2+\|v_l^h\|^2+\|v_l^{gh}\|^2}\nonumber\\
&\leq&\sqrt{\frac{(3d_f-2)(9d_gd_h^2-2d_h^2-2d_gd_h-4d_g)}{d_fd_gd_h^2}},
\end{eqnarray}
where we have used $\|A+B\|_{tr}\leq\|A\|_{tr}+\|B\|_{tr}$ for matrices $A$ and $B$ and
$\||a\rangle\langle b|\|_{tr}=\||a\rangle\|\||b\rangle\|$ for vectors $|a\rangle$ and $|b\rangle$.
\begin{flushright}
$\square$
\end{flushright}
\end{proof}

{\bf Remark 2.} We may analyze the bipartition $fg|h$ by using similar methods above. If the state $\rho\in H_1^{d_1}\otimes H_2^{d_2}\otimes H_3^{d_3}$ is separable under the bipartition $fg|l$, then
\be
\|S(\rho_{fg|h})\|_{tr}\leq\sqrt{\frac{(3d_h-2)(9d_fd_g^2-2d_g^2-2d_fd_g-4d_f)}{d_fd_g^2d_h}}.
\ee
We next consider the full separability of $\rho$.
Denote the matrix
\begin{eqnarray}
S(\rho_{f|g|h})=
\left( \begin{array}{ccccccccccccccc}
           \displaystyle(T^{(h)})^t& \displaystyle(T^{(gh)})^t  \\
           \displaystyle T^{(fh)}&\displaystyle T^{(fgh)}\\
           \end{array}
      \right ).
\end{eqnarray}
Using this matrix and the inequality for 1-body correlation tensors $\|T^{(j)}\|^2\leq \frac{2(d_j-1)}{d_j}(j=f, g, h)$\cite{vh}, we get the following separability criterion.
\begin{theorem}\label{thm:2}
If the state $\rho\in H_1^{d_1}\otimes H_2^{d_2}\otimes H_3^{d_3}$ is fully separable, then
\begin{eqnarray}
\|S(\rho_{f|g|h})\|_{tr}\leq\frac{2(d_h-1)}{d_h}\sqrt{\frac{(3d_f-2)(3d_g-2)}{d_fd_g}}.
\end{eqnarray}
\end{theorem}
\begin{proof}
A tripartite mixed state $\rho$ is fully separable whenever it can be expressed as
\begin{eqnarray}
\rho_{f|g|h}=\sum_lp_l\rho_l^{(f)}\otimes\rho_l^{(g)}\otimes\rho_l^{(h)},
\end{eqnarray}
where the probabilities $p_l>0, \sum_lp_l=1$. Let $\rho_l^{(f)}=\frac{1}{d_f}I_{d_f}+\frac{1}{2}\sum_{i_f=1}^{d_f^2-1}t_{li_f}^f\lambda_{i_f}^{(f)}$, $\rho_l^{(g)}=\frac{1}{d_g}I_{d_g}+\frac{1}{2}\sum_{i_g=1}^{d_g^2-1}t_{li_g}^g\lambda_{i_g}^{(g)}$, $\rho_l^{(h)}=\frac{1}{d_h}I_{d_h}+\frac{1}{2}\sum_{i_h=1}^{d_h^2-1}t_{li_h}^h\lambda_{i_h}^{(h)}$. Let $v_l^f, v_l^g$ and $v_l^h$ be the column vectors with entries $t_{li_f}^f, t_{li_g}^{g}$ and $t_{li_h}^{h}$ respectively. Therefore
\begin{eqnarray}
T^{(h)}=\sum_lp_lv_l^h, T^{(fh)}=\sum_lp_lv_l^f(v_l^h)^t, T^{(gh)}=\sum_lp_lv_l^g\otimes v_l^h, T^{(fgh)}=\sum_lp_lv_l^f(v_l^g\otimes v_l^h)^t.
\end{eqnarray}
It follows that the matrix $S(\rho_{f|g|h})$ can be written as
\begin{eqnarray}
S(\rho_{f|g|h})=\sum_lp_l
\left( \begin{array}{ccccccccccccccc}
            \displaystyle (v_l^h)^t&\displaystyle(v_l^g\otimes v_l^h)^t\\
             \displaystyle v_l^f (v_l^h)^t&\displaystyle v_l^f (v_l^g\otimes v_l^h)^t\\
           \end{array}
      \right )
=\sum_lp_l
\left( \begin{array}{ccccccccccccccc}
             1\\
             \displaystyle v_l^f\\
            \end{array}
     \right )
\left( \begin{array}{ccccccccccccccc}
             1&\displaystyle(v_l^g)^t\\
            \end{array}
     \right )
\otimes\ (v_l^h)^t.
\end{eqnarray}
Taking the norm, we get that
\begin{eqnarray}
\|S(\rho_{f|g|h})\|_{tr}&\leq&\sum_lp_l\|
\left( \begin{array}{ccccccccccccccc}
             1\\
             v_l^f\\
            \end{array}
     \right )
\left( \begin{array}{ccccccccccccccc}
             1&(v_l^g)^t\\
            \end{array}
     \right )
\otimes(v_l^h)^t\|_{tr}
=\sum_lp_l\|
\left( \begin{array}{ccccccccccccccc}
             1\\
             v_l^f\\
            \end{array}
     \right )\|
\|\left( \begin{array}{ccccccccccccccc}
             1\\
             v_l^g\\
            \end{array}
     \right )
\otimes v_l^h\|\nonumber\\
&=&\sum_lp_l\|
\left( \begin{array}{ccccccccccccccc}
             1\\
             v_l^f\\
            \end{array}
     \right )\|
\|\left( \begin{array}{ccccccccccccccc}
             1\\
             v_l^g\\
            \end{array}
     \right )\|
\|v_l^h\|
\leq\sum_lp_l\|v_l^h\|\sqrt{1+\|v_l^f\|^2}\sqrt{1+\|v_l^g\|^2}\nonumber\\
&\leq&\frac{2(d_h-1)}{d_h}\sqrt{\frac{(3d_f-2)(3d_g-2)}{d_fd_g}}.
\end{eqnarray}
\begin{flushright}
$\square$
\end{flushright}
\end{proof}
{\bf Remark 3.} By the above Theorem \ref{thm:1} and Theorem \ref{thm:2}, we have obtained upper bounds for 1-2, 2-1 and 1-1-1 separable quantum states. With these bounds a complete classification of tripartite quantum states has been derived.

In Ref. \cite{lww}, the authors analyzed necessary conditions of separability under arbitrary partition for four-patite quantum states by using the norms of correlation tensors. We can derive the necessary condition of separability under arbitrary partition for tripartite state $\rho\in H_1^d\otimes H_2^d\otimes H_3^d$ by using methods as similar as \cite{lww}, i.e.,
\begin{eqnarray}\label{spe}
\|T^{(123)}\|^2\leq\left\{
             \begin{array}{lr}
               \frac{8(d-1)^2(d+1)}{d^3},\ if \ \rho\ is \ 1-2 \ separable;\\
               \frac{8(d-1)^3}{d^3}, \  if \ \rho \ is \ 1-1-1 \ separable.
             \end{array}
\right.
\end{eqnarray}
\textit{\textbf{Example 1.}} Consider the quantum state $\rho\in H_1^2\otimes H_2^2\otimes H_3^2$:
\begin{eqnarray}
\rho=\frac{x}{8}I_8+(1-x)|\phi\rangle\langle\phi|,
\end{eqnarray}
where $|\phi\rangle=\frac{1}{\sqrt{2}}(|000\rangle+|111\rangle)$, $I_{8}$ stands for the $8\times8$ identity matrix and $x\in [0,1]$. By our Theorem \ref{thm:1} and Theorem \ref{thm:2}, we have $f_1(x)=\|S(\rho_{f|gh})\|_{tr}-2\sqrt{3}=-3x+4-2\sqrt{3}$ and $f_3(x)=\|S(\rho_{f|gh})\|_{tr}-2=-(3+\sqrt{2})x+\sqrt{2}+1$ respectively. When $f_1(x)>0$ or $f_2(x)>0$, $\rho$ is not separable under the bipartition $f|gh$ or fully separable. When $d=2$, according to  inequality (\ref{spe}), we have $f_2(x)=\|T^{(123)}\|^2-3=3x^2-6x$ and $f_4(x)=\|T^{(123)}\|^2-1=3x^2-6x+2$. And $\rho$ is not separable under the bipartition $f|gh$ or fully separable for $f_2(x)>0$ or $f_4(x)>0$ respectively. Fig. 1 shows that $\rho$ is not separable under the bipartition $f|gh$ for $0\leq x<0.179$ by using Theorem \ref{thm:1}, while according to  inequality (\ref{spe}), it cannot detect whether the $\rho$ is separable under the bipartition $f|gh$ or not. And $\rho$ is not fully separable for $0\leq x<0.547$ by using Theorem \ref{thm:2}, while according to inequality (\ref{spe}), $\rho$ is not fully separable for $0\leq x<0.427$. This shows that Theorem \ref{thm:1} and Theorem \ref{thm:2} detect more entangled states.

\begin{figure}
\centering
\includegraphics[width=8cm]{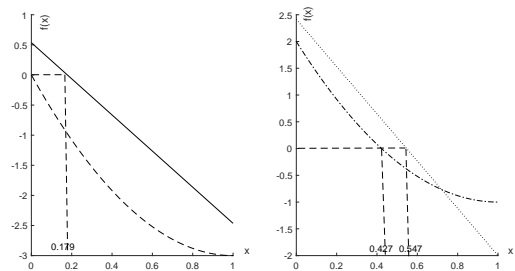}\\
\caption{The $f_1(x)$ from Theorem \ref{thm:1} (solid straight line), $f_2(x)$ from  inequality (\ref{spe}) (dashed curve line), $f_3(x)$ from Theorem \ref{thm:2} (dotted straight line) and $f_4(x)$ from  inequality (\ref{spe})(dash-dot curve line).}
\end{figure}

\section{Separability criteria for four-partite Quantum Systems}
We next consider the separability criteria for four-partite quantum states. Let $H_n^{d_n}$ $(n=1, 2, 3, 4)$ be $d_n$-dimensional Hilbert spaces. For $f=1, 2, 3, 4$, $i_f=1,\cdots,d_f^2-1$, let $\lambda_{i_f}^{(f)}$ denote the mutually orthogonal generators of the special unitary Lie algebra $\mathfrak{su}(d_f)$ under
a fixed bilinear form, and $I$ the $d_m\times d_m$, $m=1, 2, 3, 4$, identity matrix. A four-partite state $\rho\in H_1^{d_1}\otimes H_2^{d_2}\otimes H_3^{d_3}\otimes H_4^{d_4}$ can be written as follows:
\begin{eqnarray}
\rho&=&\frac{1}{d_1d_2d_3d_4}I\otimes I\otimes I\otimes I+\frac{1}{2d_gd_hd_e}\sum_{f=1}^4\sum_{i_f=1}^{d_f^2-1}t_{i_f}^f\lambda_{i_f}^{(f)}\otimes I\otimes I\otimes I+\cdots\nonumber\\
&+&\frac{1}{16}\sum_{i_1=1}^{d_1^2-1}\sum_{i_2=1}^{d_2^2-1}\sum_{i_3=1}^{d_3^2-1}\sum_{i_4=1}^{d_4^2-1}t_{i_1i_2i_3i_4}^{1234}\lambda_{i_1}^{(1)}\otimes \lambda_{i_2}^{(2)}\otimes \lambda_{i_3}^{(3)}\otimes \lambda_{i_4}^{(4)},
\end{eqnarray}
where $\lambda_{i_f}^{(f)}$ ($(f)$ indicates the position of $\lambda_{i_f}$ in the tensor product) stands for the operators on $H_{d_f}$, $I$ on the remaining appropriate factors, $t_{i_f}^f=tr(\rho\lambda_{i_f}^{(f)}\otimes I\otimes I\otimes I)$, $t_{i_fi_g}^{fg}=tr(\rho\lambda_{i_f}^{(f)}\otimes\lambda_{i_g}^{(g)}\otimes I\otimes I),  t_{i_fi_gi_h}^{fgh}=tr(\rho\lambda_{i_f}^{(f)}\otimes\lambda_{i_g}^{(g)}\otimes \lambda_{i_h}^{(h)}\otimes I)$, and $t_{i_1i_2i_3i_4}^{1234}=tr(\rho\lambda_{i_1}^{(1)}\otimes \lambda_{i_2}^{(2)}\otimes \lambda_{i_3}^{(3)}\otimes \lambda_{i_4}^{(4)})$. Let $T^{(f)}, T^{(fg)}, T^{(fgh)}, T^{(1234)}$ be the vectors (tensors) with entries $t_{i_f}^{f}, t_{i_fi_g}^{fg}, t_{i_fi_gi_h}^{fgh}, t_{i_1i_2i_3i_4}^{1234}$, respectively, where $1\leq f<g<h\leq4$, then we get $\|T^{(f)}\|^2=\sum_{i_f=1}^{d_f^2-1}(t_{i_f}^{f})^2$, $\|T^{(fg)}\|^2=\sum_{i_f=1}^{d_f^2-1}\sum_{i_g=1}^{d_g^2-1}(t_{i_fi_g}^{fg})^2$, $\|T^{(fgh)}\|^2=\sum_{i_f=1}^{d_f^2-1}\sum_{i_g=1}^{d_g^2-1}\sum_{i_h=1}^{d_h^2-1}(t_{i_fi_gi_h}^{fgh})^2$, and $\|T^{(1234)}\|^2=\sum_{i_1=1}^{d_1^2-1}\sum_{i_2=1}^{d_2^2-1}\\ \sum_{i_3=1}^{d_3^2-1}\sum_{i_4=1}^{d_4^2-1}(t_{i_1i_2i_3i_4}^{1234})^2$.
\begin{proposition}\label{prop:2}
 Let $\rho\in H_1^{d_1}\otimes\cdots\otimes H_4^{d_4}$ be a pure state, for $2\leq d_1\leq d_2\leq d_3\leq d_4$,
\begin{eqnarray}
\|T^{(1234)}\|^2\leq16(1-\frac{d_2d_3d_4+d_1d_3d_4+d_1d_2d_4+d_1d_2d_3-d_1-d_2-d_3+d_4}{2d_1d_2d_3d_4^2}).
\end{eqnarray}
\end{proposition}
\begin{proof}
Let $\rho_{i_1}$ and $\rho_{i_2i_3i_4}$ be the density matrices with respect to the subsystem $H_{d_{i_1}}$ $(i_1=1,\cdots,4)$ and $H_{d_{i_2}d_{i_3}d_{i_4}}$  $(1\leq i_2<i_3<i_4\leq4)$ respectively. For a pure state $\rho$, we have $tr(\rho^2)=1$ and $tr(\rho_{i_1}^2)=tr(\rho_{i_2i_3i_4}^2)$. Therefore
\begin{eqnarray}
\|T^{(1234)}\|^2&=&16(1-\frac{1}{d_1d_2d_3d_4})-8(\frac{1}{d_2d_3d_4}\|T^{(1)}\|^2+\cdots+\frac{1}{d_1d_2d_3}\|T^{(4)}\|^2)\nonumber\\
&-&4(\frac{1}{d_3d_4}\|T^{(12)}\|^2+\cdots+\frac{1}{d_1d_2}\|T^{(34)}\|^2)-2(\frac{1}{d_4}\|T^{(123)}\|^2+\cdots+\frac{1}{d_1}\|T^{(234)}\|^2)\nonumber\\
&\leq&16(1-\frac{1}{d_1d_2d_3d_4})-8(\frac{1}{d_2d_3d_4}\|T^{(1)}\|^2+\cdots+\frac{1}{d_1d_2d_3}\|T^{(4)}\|^2)\nonumber\\
&-&\frac{4}{d_4}(\frac{1}{d_3}\|T^{(12)}\|^2+\cdots+\frac{1}{d_1}\|T^{(34)}\|^2)-2(\frac{1}{d_4}\|T^{(123)}\|^2+\cdots+\frac{1}{d_1}\|T^{(234)}\|^2)\nonumber
\end{eqnarray}
\begin{eqnarray}
&\leq&16(1-\frac{1}{d_1d_2d_3d_4})-8(\frac{1}{d_2d_3d_4}\|T^{(1)}\|^2+\cdots+\frac{1}{d_1d_2d_3}\|T^{(4)}\|^2)\nonumber\\
&-&\frac{8}{d_4}[\frac{d_2d_3d_4+d_1d_3d_4+d_1d_2d_4+d_1d_2d_3-d_1-d_2-d_3-d_4}{d_1d_2d_3d_4}\nonumber\\
&+&\frac{1}{2}(\frac{d_2d_3d_4-d_2-d_3-d_4}{d_2d_3d_4}\|T^{(1)}\|^2+\cdots+\frac{d_1d_2d_3-d_1-d_2-d_3}{d_1d_2d_3}\|T^{(4)}\|^2)\nonumber\\
&-&\frac{1}{2}(\|T^{(123)}\|^2+\cdots+\|T^{(234)}\|^2)]-2(\frac{1}{d_4}\|T^{(123)}\|^2+\cdots+\frac{1}{d_1}\|T^{(234)}\|^2)\nonumber\\
&\leq&16(1-\frac{d_2d_3d_4+d_1d_3d_4+d_1d_2d_4+d_1d_2d_3-d_1-d_2-d_3+d_4}{2d_1d_2d_3d_4^2})\nonumber\\
&-&8(\frac{d_2d_3d_4-d_2-d_3+d_4}{2d_2d_3d_4^2}\|T^{(1)}\|^2
+\frac{d_1d_3d_4-d_1-d_3+d_4}{2d_1d_3d_4^2}\|T^{(2)}\|^2\nonumber\\
&+&\frac{d_1d_2d_4-d_1-d_2+d_4}{2d_1d_2d_4^2}\|T^{(3)}\|^2
+\frac{2d_4+d_1d_2d_3-d_1-d_2-d_3}{2d_1d_2d_3d_4}\|T^{(4)}\|^2)\nonumber\\
&-&2(\frac{1}{2d_4}\|T^{(123)}\|^2+\frac{2d_4-d_3}{2d_3d_4}\|T^{(124)}\|^2+\frac{2d_4-d_2}{2d_2d_4}\|T^{(134)}\|^2
+\frac{2d_4-d_1}{2d_1d_4}\|T^{(234)}\|^2)\nonumber\\
&\leq&16(1-\frac{d_2d_3d_4+d_1d_3d_4+d_1d_2d_4+d_1d_2d_3-d_1-d_2-d_3+d_4}{2d_1d_2d_3d_4^2}).
\end{eqnarray}
\begin{flushright}
$\square$
\end{flushright}
\end{proof}
{\bf Remark 4.} When $d_1=d_2=d_3=d_4=d$, we can obtain $\|T^{(1234)}\|\leq\frac{16(d^2-1)^2}{d^4}$. Thus, Proposition \ref{prop:2} generalizes Theorem 2 in \cite{lww}.

For a four-partite quantum state $\rho$ on $H_1^{d_1}\otimes H_2^{d_2}\otimes H_3^{d_3}\otimes H_4^{d_4}$, the bipartitions and tripartitions are the following: $f|ghe,fg|he,fgh|e$ and $f|g|he, fg|h|e, f|gh|e$ where $f\neq g\neq h\neq e=1,2,3,4$. We first consider the separability of $\rho$ under the bipartition $f|ghe$.
Define the matrix  $S(\rho_{f|ghe})$ by
\be
S(\rho_{f|ghe})=\left(
\begin{array}{cc}
            1&(T^{(g)})^t\\
            T^{(f)}&T^{(fg)}\\
\end{array}
\right ).
\ee
And by the inequality $\|T^{(j)}\|^2\leq \frac{2(d_j-1)}{d_j}(j=f, g)$\cite{vh}, we get the following separability criterion.
\begin{theorem}\label{thm:3}
If the state $\rho\in H_1^{d_1}\otimes H_2^{d_2}\otimes H_3^{d_3}\otimes H_4^{d_4}$ is separable under the bipartition $f|ghe$, then
\begin{eqnarray}
\|S(\rho_{f|ghe})\|_{tr}\leq\sqrt{\frac{(3d_f-2)(3d_g-2)}{d_fd_g}}.
\end{eqnarray}
\end{theorem}

\begin{proof}A four-partite mixed state $\rho$ is separable under bipartition $f|ghe$ whenever it can be expressed as
\begin{eqnarray}
\rho_{f|ghe}=\sum_lp_l\rho_l^{(f)}\otimes\rho_l^{(ghe)},
\end{eqnarray}
where the probabilities $p_l>0, \sum_lp_l=1$. Let $\rho_l^{(f)}=\frac{1}{d_f}I_{d_f}+\frac{1}{2}\sum_{i_f=1}^{d_f^2-1} t_{li_f}^f\lambda_{i_f}^{(f)}, \rho_l^{(ghe)}=\frac{1}{d_gd_hd_e}I_{d_g}\otimes I_{d_h}\otimes I_{d_e}+\frac{1}{2}(\frac{1}{d_hd_e}\sum_{i_g=1}^{d_g^2-1} t_{li_g}^g\lambda_{i_g}^{(g)}\otimes I_{d_h}\otimes I_{d_e}+\cdots)+\frac{1}{4}(\frac{1}{d_e}\sum_{i_g=1}^{d_g^2-1}\sum_{i_h=1}^{d_h^2-1} t_{li_gi_h}^{gh}\lambda_{i_g}^{(g)}\otimes\lambda_{i_h}^{(h)}\otimes I_{d_e}+\cdots)+\frac{1}{8}\sum_{i_g=1}^{d_g^2-1}\sum_{i_h=1}^{d_h^2-1}\sum_{i_e=1}^{d_e^2-1} t_{li_gi_hi_e}^{ghe}\lambda_{i_g}^{(g)}\otimes\lambda_{i_h}^{(h)}\otimes \lambda_{i_e}^{(e)}$. Let $v_l^f$ and $v_l^g$ be the column vectors with entries $t_{li_f}^f$ and $t_{li_g}^g$ respectively. Therefore
\begin{eqnarray}
&&T^{(f)}=\sum_lp_lv_l^f, T^{(g)}=\sum_lp_lv_l^g, T^{(fg)}=\sum_lp_lv_l^f(v_l^g)^t,
\end{eqnarray}
and then it follows that the matrix $S(\rho_{f|ghe})$ can be written as
\begin{eqnarray}
S(\rho_{f|ghe})=\sum_lp_l
\left( \begin{array}{cc}
            1&(v_l^g)^t\\
            v_l^f&v_l^f(v_l^g)^t\\
           \end{array}
      \right )
=\sum_lp_l
\left( \begin{array}{cc}
             1\\
             v_l^f\\
            \end{array}
     \right )
\left( \begin{array}{ccccccccccccccc}
              1&(v_l^g)^t\\
            \end{array}
     \right ).
\end{eqnarray}
Thus,
\begin{eqnarray}
\|S(\rho_{f|ghe})\|_{tr} &\leq&\sum_lp_l\|
\left( \begin{array}{c}
             1\\
             v_l^f\\
            \end{array}
     \right )
\left( \begin{array}{cc}
              1&(v_l^g)^t\\
            \end{array}
     \right )\|_{tr}
=\sum_lp_l\|
\left( \begin{array}{c}
             1\\
             v_l^f\\
            \end{array}
     \right )\|
\left( \begin{array}{cc}
              1&(v_l^g)^t\\
            \end{array}
     \right )^t\|\nonumber\\
&=&\sum_lp_l\sqrt{1+\|v_l^f\|^2}\sqrt{1+\|v_l^g\|^2}\leq\sqrt{\frac{(3d_f-2)(3d_g-2)}{d_fd_g}}.
\end{eqnarray}
\begin{flushright}
$\square$
\end{flushright}
\end{proof}

We next consider the separability of $\rho$ under bipartition $fg|he$.
Denote the matrix 
\begin{eqnarray}
S(\rho_{fg|he})=
\left( \begin{array}{cc}
           1&(T^{(he)})^t \\
           T^{(fg)}&T^{(fghe)}\\
           \end{array}
      \right ).
\end{eqnarray}
And by Lemma \ref{lemma:1}, we have 
the following separability criterion.
\begin{theorem}\label{thm:4}
If the state $\rho\in H_1^{d_1}\otimes H_2^{d_2}\otimes H_3^{d_3}\otimes H_4^{d_4}$ is separable under bipartition $fg|he$, then
\begin{eqnarray}
\|S(\rho_{fg|he})\|_{tr}\leq\frac{\sqrt{(5d_g^2-4)(5d_e^2-4)}}{d_gd_e}.
\end{eqnarray}
\end{theorem}
\begin{proof} A four-partite mixed state $\rho$ is separable under bipartition $fg|he$ whenever it can be expressed as
\begin{eqnarray}
\rho_{fg|he}=\sum_lp_l\rho_l^{(fg)}\otimes\rho_l^{(he)},
\end{eqnarray}
where the probabilities $p_l>0, \sum_lp_l=1$.

Let $\rho_l^{(fg)}=\frac{1}{d_fd_g}I_{d_f}\otimes I_{d_g}+\cdots+\frac{1}{4}\sum_{i_f=1}^{d_f^2-1}\sum_{i_g=1}^{d_g^2-1} t_{li_fi_g}^{fg}\lambda_{i_f}^{(f)}\otimes\lambda_{i_g}^{(g)}$, and $\rho_l^{(he)}=\frac{1}{d_hd_e}I_{d_h}\otimes I_{d_e}+\cdots+\frac{1}{4}\sum_{i_h=1}^{d_h^2-1}\sum_{i_e=1}^{d_e^2-1}t_{li_hi_e}^{he}\lambda_{i_h}^{(h)}\otimes\lambda_{i_e}^{(e)}$.
Let $v_l^{fg}$ and $v_l^{he}$ be the column vectors with entries $t_{li_fi_g}^{fg}$ and $t_{li_hi_e}^{he}$ respectively. Therefore
\begin{eqnarray}
&&T^{(fg)}=\sum_lp_lv_l^{fg}, T^{(he)}=\sum_lp_lv_l^{he}, T^{(fghe)}=\sum_lp_lv_l^{fg}(v_l^{he})^t.
\end{eqnarray}
Then the matrix matrix $S(\rho_{fg|he})$ can be written as
\begin{eqnarray}
S(\rho_{fg|he})&=&\sum_lp_l
\left( \begin{array}{cc}
          1&(v_l^{he})^t \\
          v_l^{fg}&v_l^{fg}(v_l^{he})^t\\
           \end{array}
      \right )
=\sum_lp_l
\left( \begin{array}{cc}
             1&(v_l^{fg})^t\\
            \end{array}
     \right )^t
\left( \begin{array}{cc}
             1&(v_l^{he})^t\\
            \end{array}
     \right ).
\end{eqnarray}
Thus,
\begin{eqnarray}
\|S(\rho_{fg|he})\|_{tr}&\leq&\sum_lp_l\|
\left( \begin{array}{cc}
             1&(v_l^{fg})^t\\
            \end{array}
     \right )^t
\left( \begin{array}{cc}
             1&(v_l^{he})^t\\
            \end{array}
     \right )\|_{tr}\nonumber\\
&=&\sum_lp_l\|
\left( \begin{array}{cc}
        1&(v_l^{fg})^t\\
            \end{array}
     \right )^t\|
\|\left( \begin{array}{cc}
          1&(v_l^{he})^t\\
            \end{array}
     \right )^t\|\nonumber\\
&=&\sum_lp_l\sqrt{1+\|v_l^{fg}\|^2}\sqrt{1+\|v_l^{he}\|^2}\nonumber\\
&\leq&\frac{\sqrt{(5d_g^2-4)(5d_e^2-4)}}{d_gd_e},
\end{eqnarray}
where we have used the triangular inequality of the trace norm.
\begin{flushright}
$\square$
\end{flushright}
\end{proof}

We next consider the separability of $\rho$ under tripartition $f|g|he$. We define the matrix $S(\rho_{f|g|he})$ by
\begin{eqnarray}
S(\rho_{f|g|he})=
\left( \begin{array}{cc}
           (T^{(h)})^t& T^{(gh)^t}  \\
           T^{(fh)}&T^{(fgh)}\\
           \end{array}
      \right ),
\end{eqnarray}
with the matrix, the inequality $\|T^{(j)}\|^2\leq \frac{2(d_j-1)}{d_j}(j=f, g, h)$\cite{vh} and similar methods of Theorem \ref{thm:2}, then we obtain the following separability criterion.
\begin{theorem}\label{thm:5}
If the state $\rho\in H_1^{d_1}\otimes H_2^{d_2}\otimes H_3^{d_3}\otimes H_4^{d_4}$ is separable under tripartition $f|g|he$, then
\begin{eqnarray}
\|S(\rho_{f|g|he})\|_{tr}\leq\frac{2(d_h-1)}{d_h}\sqrt{\frac{(3d_f-2)(3d_g-2)}{d_fd_g}}.
\end{eqnarray}
\end{theorem}
{\bf Remark 5.} We may analyze the bipartition $fgh|e$ and tripartitions $fg|h|e$, $f|gh|e$ by using similar methods of Theorem \ref{thm:2} and Theorem \ref{thm:4} above, respectively. If the state $\rho\in H_1^{d_1}\otimes H_2^{d_2}\otimes H_3^{d_3}\otimes H_4^{d_4}$ is separable under bipartition $fgh|e$, then
$\|S_(\rho_{fgh|e})\|_{tr}\leq\sqrt{\frac{(3d_f-2)(3d_e-2)}{d_fd_e}}.$ If the state $\rho\in H_1^{d_1}\otimes H_2^{d_2}\otimes H_3^{d_3}\otimes H_4^{d_4}$ is separable under tripartition $fg|h|e$ or $f|gh|e$, then
$\|S(\rho_{fg|h|e})\|_{tr}\leq\frac{2(d_f-1)}{d_f}\sqrt{\frac{(3d_h-2)(3d_e-2)}{d_hd_e}}$ or $\|S(\rho_{f|gh|e})\|_{tr}\leq\frac{2(d_g-1)}{d_g}\sqrt{\frac{(3d_f-2)(3d_e-2)}{d_fd_e}}$, respectively.

We next consider the full separability of $\rho$.
Denote the matrix
\begin{eqnarray}
S(\rho_{f|g|h|e})=
\left( \begin{array}{cccc}
           (T^{(e)})^t& (T^{(he)})^t& (T^{(ge)})^t&(T^{(ghe)})^t \\
           T^{(fe)}&T^{(fhe)}&T^{(fge)}&T^{(fghe)}\\
           \end{array}
      \right ),
\end{eqnarray}
Using this matrix and the inequality for 1-body correlation tensors $\|T^{(j)}\|^2\leq \frac{2(d_j-1)}{d_j}(j=f, g, h, e)$\cite{vh}, we get the following separability criterion.
\begin{theorem}\label{thm:6}
If the state $\rho\in H_1^{d_1}\otimes H_2^{d_2}\otimes H_3^{d_3}\otimes H_4^{d_4}$ is fully separable, then
\begin{eqnarray}
\|S(\rho_{f|g|h|e})\|_{tr}\leq\frac{2(d_e-1)}{d_e}\sqrt{\frac{(3d_f-2)(3d_g-2)(3d_h-2)}{d_fd_gd_h}}.
\end{eqnarray}
\end{theorem}
\begin{proof}
A four partite mixed state $\rho$ is fully separable whenever it can be expressed as
\begin{eqnarray}
\rho_{f|g|h|e}=\sum_lp_l\rho_l^{(f)}\otimes\rho_l^{(g)}\otimes\rho_l^{(h)}\otimes\rho_l^{(e)},
\end{eqnarray}
where the probabilities $p_l>0, \sum_lp_l=1$. Let $\rho_l^{(f)}=\frac{1}{d_f}I_{d_f}+\frac{1}{2}\sum_{i_f=1}^{d_f^2-1}t_{li_f}^f\lambda_{i_f}^{(f)}$, $\rho_l^{(g)}=\frac{1}{d_g}I_{d_g}+\frac{1}{2}\sum_{i_g=1}^{d_g^2-1}t_{li_g}^g\lambda_{i_g}^{(g)}$, $\rho_l^{(h)}=\frac{1}{d_h}I_{d_h}+\frac{1}{2}\sum_{i_h=1}^{d_h^2-1}t_{li_h}^h\lambda_{i_h}^{(h)}$, $\rho_l^{(e)}=\frac{1}{d_e}I_{d_e}+\frac{1}{2}\sum_{i_e=1}^{d_e^2-1}t_{li_e}^e\lambda_{i_e}^{(e)}$. Let $v_l^f, v_l^g, v_l^h$ and $v_l^e$ be the column vectors with entries $t_{li_f}^f, t_{li_g}^{g}, t_{li_h}^{h}$ and $t_{li_e}^{e}$ respectively. Therefore,
\begin{eqnarray}
&& T^{(e)}=\sum_lp_lv_l^e,\ T^{(fe)}=\sum_lp_lv_l^f(v_l^e)^t, \ T^{(ge)}=\sum_lp_lv_l^g\otimes v_l^e, \ T^{(he)}=\sum_lp_lv_l^h\otimes v_l^e, \nonumber\\ &&T^{(fge)}=\sum_lp_lv_l^f(v_l^g\otimes v_l^e)^t, \ T^{(ghe)}=\sum_lp_lv_l^g\otimes v_l^h\otimes v_l^e,\ T^{(fhe)}=\sum_lp_lv_l^f(v_l^h\otimes v_l^e)^t,\nonumber\\
&& T^{(fghe)}=\sum_lp_lv_l^f(v_l^g\otimes v_l^h\otimes v_l^e)^t,
\end{eqnarray}
where $t$ stands for transpose.
Then the matrix $S(\rho_{f|g|h|e})$ can be written as
\begin{eqnarray}
S(\rho_{f|g|h|e})&=&\sum_lp_l
\left( \begin{array}{ccccccccccccccc}
            (v_l^e)^t& (v_l^h\otimes v_l^e)^t& (v_l^g\otimes v_l^e)^t& (v_l^g\otimes v_l^h\otimes v_l^e)^t \\
            v_l^f(v_l^e)^t & v_l^f(v_l^h\otimes v_l^e)^t& v_l^f(v_l^g\otimes v_l^e)^t&v_l^f(v_l^g\otimes v_l^h\otimes v_l^e)^t\\
           \end{array}
      \right )
\nonumber\\
&=&\sum_lp_l
\left( \begin{array}{c}
              1\\
              v_l^f\\
            \end{array}
     \right )
\left( \begin{array}{cc}
              1&(v_l^g)^t\\
            \end{array}
     \right )\otimes
\left( \begin{array}{cc}
              1&(v_l^h)^t\\
            \end{array}
     \right )\otimes (v_l^e)^t.
\end{eqnarray}
Thus,
\begin{eqnarray}
\|S(\rho_{f|g|h|e})\|_{tr}&\leq&\sum_lp_l\|
\left( \begin{array}{c}
              1\\
              v_l^f\\
            \end{array}
     \right )
\left( \begin{array}{cc}
              1&(v_l^g)^t\\
            \end{array}
     \right )\otimes
\left( \begin{array}{cc}
              1&(v_l^h)^t\\
            \end{array}
     \right )\otimes (v_l^e)^t
\|_{tr}\nonumber\\
&=&\sum_lp_l\|
\left( \begin{array}{c}
                 1\\
               v_l^f\\
            \end{array}
     \right )\|
\|\left( \begin{array}{c}
                 1\\
               v_l^g\\
            \end{array}
     \right )\otimes
\left( \begin{array}{c}
                 1\\
               v_l^h\\
            \end{array}
     \right )\otimes v_l^e\|\nonumber\\
&=&\sum_lp_l\|
\left( \begin{array}{c}
                 1\\
               v_l^f\\
            \end{array}
     \right )\|
\|\left( \begin{array}{c}
                 1\\
               v_l^g\\
            \end{array}
     \right )\|
\|\left( \begin{array}{c}
                 1\\
               v_l^h\\
            \end{array}
     \right )\|
\|v_l^e\|\nonumber\\
&=&\sum_lp_l \|v_l^e\|\sqrt{1+\|v_l^f\|^2}\sqrt{1+\|v_l^g\|^2}\sqrt{1+\|v_l^h\|^2}\nonumber\\
&\leq&\frac{2(d_e-1)}{d_e}\sqrt{\frac{(3d_f-2)(3d_g-2)(3d_h-2)}{d_fd_gd_h}}.
\end{eqnarray}
\begin{flushright}
$\square$
\end{flushright}
\end{proof}
{\bf Remark 6.} From Theorem \ref{thm:3} to Theorem \ref{thm:6}, we have derived the upper bounds for 1-3, 3-1, 2-2, 1-1-2, 1-2-1, 2-1-1 and 1-1-1-1 separable quantum states. Thus, we can obtain a complete classification of four-partite quantum states with these bounds.

\textit{\textbf{Example 2.}} Consider the quantum state $\rho\in H_1^2\otimes H_2^2\otimes H_3^2\otimes H_4^2$,
\begin{eqnarray}
\rho=x|\psi\rangle\langle\psi|+\frac{1-x}{16}I_{16},
\end{eqnarray}
where $|\psi\rangle=\frac{1}{\sqrt{2}}(|0000\rangle+|1111\rangle)$ and $I_{16}$ stands for the $16\times 16$ identity matrix. By Theorem \ref{thm:4}, we have $f_1(x)=\|S(\rho_{fg|he})\|_{tr}-4=\sqrt{1+x^2}+2\sqrt{2}x+\frac{x-x^2}{1+x^2}-4$ and when $f_1(x)>0$, $\rho$ is not separable under bipartition $fg|he$. When $d=2$, from Ref. \cite{lww}, one has that $f_2(x)=\|T^{(1234)}\|^2-\frac{16}{d^4}(d^2-1)^2=9x^2-9$ and $\rho$ is not separable under bipartition $fg|he$ for $f_2(x)>0$. Fig.2 shows that $\rho$ is not separable under bipartition $fg|he$ for $0.915< x\leq 1$ by Theorem \ref{thm:4}, while using Theorem 3 in Ref. \cite{lww}, it cannot detect whether the $\rho$ is inseparable under bipartition $fg|he$. Thus, our method detects more entangled states than that of Ref. \cite{lww}.

\begin{figure}
\centering
\includegraphics[width=8cm]{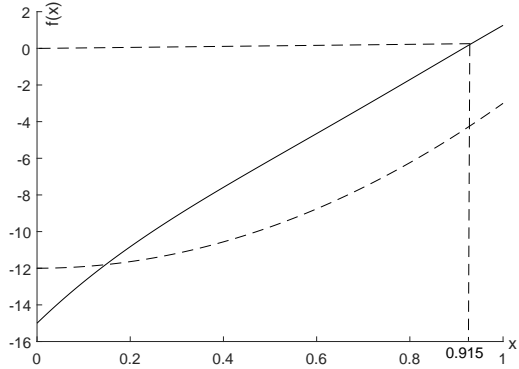}\\
\caption{The function $f_1(x)$ from Theorem \ref{thm:4} (solid line) and $f_2(x)$ from Ref. \cite{lww} (dashed line).}
\end{figure}

\textit{\textbf{Example 3.}} Consider the quantum state $\rho\in H_1^2\otimes H_2^2\otimes H_3^2\otimes H_4^2$,
\begin{eqnarray}
\rho=x|\varphi\rangle\langle\varphi|+\frac{1-x}{16}I_{16},
\end{eqnarray}
where $|\varphi\rangle=\frac{1}{2}(|0001\rangle+|0010\rangle+|0100\rangle+|1000\rangle)$ and $I_{16}$ stands for the $16\times16$ identity matrix.
By Theorem \ref{thm:3} and Theorem \ref{thm:5}, we have $f_1(x)=\|S_(\rho_{f|ghe})\|_{tr}-2=\frac{4+2x^2}{2\sqrt{4+x^2}}+x-2$ and $f_3(x)=\|S(\rho_{f|g|he})\|_{tr}-2=\frac{6+3\sqrt{2}}{4}x-2$ respectively. When $f_1(x)>0$ or $f_3(x)>0$, $\rho$ is not separable under bipartition $f|ghe$ or tripartition $f|g|he$ respectively. When $d=2$, we have $f_2(x)=\|T^{(1234)}\|^2-4=4x^2-4$ and $f_4(x)=\|T^{(1234)}\|^2-3=4x^2-3$ from Theorem 3 in Ref. \cite{lww}. And $\rho$ is not separable under bipartition $f|ghe$ or tripartition $f|g|he$ for $f_2(x)>0$ or $f_4(x)>0$ respectively. From Fig. 3, $\rho$ is not separable under bipartition $f|ghe$ for $0.783< x\leq 1$ by Theorem \ref{thm:3}, while using the Theorem 3 in Ref. \cite{lww}, it can not detect whether the $\rho$ is not separable under bipartition $f|ghe$. And $\rho$ is not separable under tripartition $f|g|he$ for $0.781< x\leq 1$ by Theorem \ref{thm:5}, while using the method in Ref. \cite{lww}, $\rho$ is not separable under tripartition $f|g|he$ for $0.866< x\leq 1$. Therefore Theorem \ref{thm:3} and Theorem \ref{thm:5} detect more entangled states than Theorem 3 in Ref. \cite{lww}.

\begin{figure}[!htb]
\centering
\includegraphics[width=8cm]{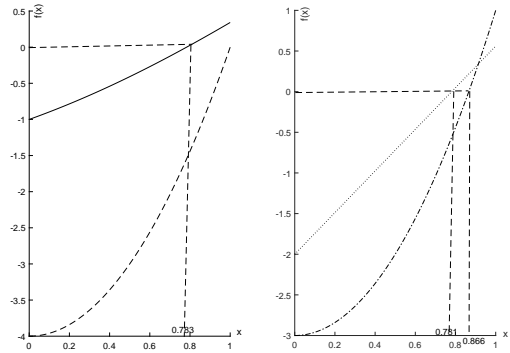}\\
\caption{The function $f_1(x)$ from Theorem \ref{thm:3} (solid curve line), $f_2(x)$ from Theorem 3 in Ref. \cite{lww} (dashed curve line), $f_3(x)$ from Theorem \ref{thm:5} (dotted straight line) and $f_4(x)$ from Theorem 3 in Ref. \cite{lww} (dash-dot curve line).}
\end{figure}

\section{Separability criteria for multipartite Quantum Systems}
We finally consider the separability criteria for $n$-partite quantum states. Let $H_i^{d}$ ($i=1, \cdots, n$) denote $d$-dimensional Hilbert spaces. Let $\lambda_{i_f}^{(f)}\,,f=1,\cdots, n$, $i_f=1,\cdots,d^2-1$ be the mutually orthogonal generators of the special unitary Lie algebra $\mathfrak{su}(d)$ under
a fixed bilinear form, and $I$ the $d\times d$ the identity matrix. A $n$-partite state $\rho\in H_1^d\otimes\ H_2^d\otimes\cdots\otimes H_n^d$ can be written as follows:
\begin{eqnarray}
\rho&=&\frac{1}{d^n}I\otimes\cdots\otimes I+\frac{1}{2d^{n-1}}\sum_{f=1}^n\sum_{i_1=1}^{d^2-1}t_{i_1}^{f}\lambda_{i_1}^{(f)}\otimes I\otimes\cdots\otimes I+\cdots\nonumber\\
&+&\frac{1}{2^n}\sum_{i_1,\cdots,i_n=1}^{d^2-1}t_{i_1,\cdots,i_n}^{1\cdots n}\lambda_{i_1}^{(1)}\otimes\lambda_{i_2}^{(2)}\otimes\cdots\otimes\lambda_{i_n}^{(n)},
\end{eqnarray}
where $\lambda_{i_f}^{(f)}$ ($(f)$ refers the position of $\lambda_{i_f}$ in the tensor product) stands for the operators with $\lambda_{i_f}$ on $H_{d_f}$, and $I$ on the remaining spaces, $t_{i_1}^{f}=tr(\rho\lambda_{i_1}^{(f)}\otimes I\otimes\cdots\otimes I),\cdots,$ $t_{i_1,\cdots,i_n}^{1\cdots n}=tr(\rho\lambda_{i_1}^{(1)}\otimes\lambda_{i_2}^{(2)}\otimes\cdots\otimes\lambda_{i_n}^{(n)})$. Let $T^{(f)},\cdots, T^{(1\cdots n)}$ be the vectors (tensors) with entries $t_{i_f}^{f}, \cdots, t_{i_1\cdots i_n}^{1\cdots n}$ respectively where $1\leq f\leq n$, then we get $\|T^{(f)}\|^2=\sum_{i_f=1}^{d^2-1}(t_{i_f}^{f})^2,\cdots, \|T^{(1\cdots n)}\|^2=\sum_{i_1,\cdots,i_n=1}^{d^2-1}(t_{i_1\cdots i_n}^{1\cdots n})^2$. Define further $A_1=\sum_{f=1}^n\|T^{(f)}\|^2, \cdots, A_n=\|T^{(1\cdots n)}\|^2$.
\begin{lemma}\label{lemma:2}
 Let $\rho\in H_1^d\otimes\ H_2^d\otimes\cdots\otimes H_n^d$ be a pure state,
\begin{eqnarray}
\|T^{(1\cdots n)}\|^2\leq(\frac{2}{d})^n\frac{(n-2)d^n-nd^{n-2}+2}{n-2}.
\end{eqnarray}
\end{lemma}
\begin{proof}
Let $\rho_{i_1}$ and $\rho_{i_2\cdots i_n}$ be the density matrices with respect to the subsystem $H_{d_{i_1}}, i_1=1,\cdots,n$, and $H_{d_{i_2}\cdots d_{i_n}}, 1\leq i_2<\cdots<i_n\leq n$. As for a pure state $\rho$ we have $tr(\rho^2)=1$ and $tr(\rho_{i_1}^2)=tr(\rho_{i_2\cdots i_n}^2)$. Therefore when $n=2$, we get $\rho_1^2=\rho_2^2$ and  $\|T^{(12)}\|\leq\frac{4(d^2-1)}{d^2}$ \cite{lww}. When $n>2$, we have
\begin{eqnarray}
A_n&=&2^n(1-\frac{1}{d^n})-\frac{2^{n-1}}{d^{n-1}}A_1-\frac{2^{n-2}}{d^{n-2}}A_2-\cdots-\frac{2}{d}A_{n-1}\nonumber\\
&=&2^n(1-\frac{1}{d^n})-\frac{2^{n-1}}{d^{n-1}}A_1-\frac{2^{n-2}}{d^{n-2}}[\frac{n(d^{n-2}-1)}{(n-2)d^n}+\frac{d^{n-2}-n+1}{2(n-2)d^{n-2}}A_1-\frac{n-3}{8(n-2)d^{n-3}}A_3-\cdots\nonumber\\
&-&\frac{1}{d(n-2)2^{n-1}}A_{n-1}]-\cdots-\frac{2}{d}A_{n-1}\nonumber
\end{eqnarray}
\begin{eqnarray}
&=&(\frac{2}{d})^n\frac{(n-2)d^n-nd^{n-2}+2}{n-2}-\frac{2^n(d^{n-2}-1)}{2(n-2)d^{n-1}}A_1-\frac{2^n}{8(n-2)d^{n-3}}A_3\nonumber\\
&-&\frac{2^{n+1}}{16(n-2)d^{n-4}}A_4-\cdots-\frac{2^n(n-3)}{d(n-2)2^{n-1}}A_{n-1}\nonumber\\
&\leq&(\frac{2}{d})^n\frac{(n-2)d^n-nd^{n-2}+2}{n-2}.
\end{eqnarray}
\begin{flushright}
$\square$
\end{flushright}
\end{proof}
{\bf Remark 7.} When the dimensions of each system for tripartite and four-partite quantum states are the same, Lemma \ref{lemma:2} specializes to Proposition \ref{prop:1} and Proposition \ref{prop:2}.

For the $n$-partite quantum state $\rho\in H_1^d\otimes\ H_2^d\otimes\cdots\otimes H_n^d$, we denote the general $k$-partite decompositions of $\rho$ as follows: $\{a^1\},\cdots,\{a^{k_1}\},\{c_1^1, c_2^1\},\cdots,\{c_1^{k_2}, c_2^{k_2}\},\cdots,\{e_1^{k_j},\cdots,e_j^{k_j}\}$, and $\sum_{m=1}^jk_m=k, \sum_{m=1}^jmk_m=n$. Denote the upper bounds of the $j$-body corralation tensors associated to partition
$(12\cdots j)$ by $w_{12\cdots j}$, namely, $\|T^{(1)}\|^2\leq\frac{2(d-1)}{d}=w_1$, $\|T^{(12)}\|^2\leq\frac{4(d^2-1)}{d^2}=w_{12}$, and $\|T^{(12\cdots j)}\|^2\leq(\frac{2}{d})^j\frac{(j-2)d^j-jd^{j-2}+2}{j-2}=w_{12\cdots j}$ (cf. Lemma \ref{lemma:2}).

\begin{theorem}\label{thm:7}
Let $\rho\in H_1^d\otimes\ H_2^d\otimes\cdots\otimes H_n^d$ be an $n$-partite $k$-separable quantum state. We have
\begin{eqnarray}
\|T^{(1\cdots n)}\|\leq(w_1)^{k_1}(w_{12})^{k_2}\cdots(w_{12\cdots j})^{k_j},
\end{eqnarray}
where $\sum_{m=1}^jk_m=k$, $\sum_{m=1}^jmk_m=n$.
\end{theorem}
\begin{proof} Assume that $|\varphi\rangle\in H_1^d\otimes\ H_2^d\otimes\cdots\otimes H_n^d$ is a pure state, say
$|\varphi\rangle=|a^1\rangle\otimes \cdots\otimes |a^{k_1}\rangle\otimes \cdots\otimes |e_1^{k_j}\cdots e_j^{k_j}\rangle$, then we have
\begin{eqnarray}
t_{i_1\cdots i_n}^{1\cdots n}&&=tr(|\varphi\rangle\langle\varphi|\lambda_{i_1}^{(1)}\otimes\lambda_{i_2}^{(2)}\otimes\cdots\otimes\lambda_{i_n}^{(n)})\nonumber\\
&&=tr(|a^1\rangle\langle a^1|\lambda_{i_1}^{(1)})\cdots tr(|a^{k_1}\rangle\langle a^{k_1}|\lambda_{i_{k_1}}^{({k_1})})\cdots tr(|e_1^{k_j}\cdots e_j^{k_j}\rangle\langle e_1^{k_j}\cdots ej^{k_j}|\nonumber\\
&&\lambda_{i_{k_1+2k_2+\cdots+(j-1)k_{j-1}+1}}^{({k_1+2k_2+\cdots+(j-1)k_{j-1}+1})}\otimes \cdots\otimes\lambda_{i_{k_1+2k_2+\cdots+jk_j}}^{({k_1+2k_2+\cdots+jk_j})})\nonumber\\
&&=t_{i_1}^1\cdots t_{i_{k_1}}^{k_1}\cdots t_{i_{k_1+2k_2+\cdots+(j-1)k_{j-1}+1},\cdots, i_{k_1+2k_2+\cdots+jk_j}}^{{k_1+2k_2+\cdots+(j-1)k_{j-1}+1},\cdots,{k_1+2k_2+\cdots+jk_j}}.
\end{eqnarray}
Thus,
\begin{eqnarray}
\|T^{(1\cdots n)}\|^2&=&\sum_{i_1,\cdots,i_n=1}^{d^2-1}(t_{i_1\cdots i_n}^{1\cdots n})^2\nonumber\\
&=&\|T^{(1)}\|^2\cdots\|T^{(k_1)}\|^2\cdots\|T^{({k_1+2k_2+\cdots+(j-1)k_{j-1}+1},\cdots,{k_1+2k_2+\cdots+jk_j})}\|\nonumber\\
&\leq&(w_1)^{k_1}(w_{12})^{k_2}\cdots(w_{12\cdots j})^{k_j}.
\end{eqnarray}
In general for any mixed state $\rho\in H_1^d\otimes\ H_2^d\otimes\cdots\otimes H_n^d$ with ensemble representation $\rho=\sum_ip_i|\varphi_i\rangle\langle\varphi_i|$ where $\sum_ip_i=1$, we derive that
\begin{eqnarray}
\|T^{(1\cdots n)}(\rho)\|^2=\|\sum_ip_iT^{(1\cdots n)}(|\varphi_i\rangle)\|^2&\leq& \sum_ip_i\|T^{(1\cdots n)}(|\varphi_i\rangle)\|^2\nonumber\\
&\leq&(w_1)^{k_1}(w_{12})^{k_2}\cdots(w_{12\cdots j})^{k_j}.
\end{eqnarray}
\begin{flushright}
$\square$
\end{flushright}
\end{proof}
{\bf Remark 8.} When $\rho\in H_1^d\otimes H_2^d\otimes H_3^d\otimes H_4^d$ be a four-partite quantum state. Explicitly by Theorem \ref{thm:7}, we have that
\be
\|T^{(1234)}\|\leq\left\{
             \begin{array}{lr}
           \frac{16}{d^4}(d-1)(d^3-3d+2), \ if \ \rho \ is \ 1-3 \ separable;\\
           \frac{16}{d^4}(d^2-1)^2, \ if \ \rho \ is \ 2-2 \ separable;\\
           \frac{16}{d^4}(d^2-1)(d-1)^2, \ if \ \rho \ is \ 1-1-2 \ separable;\\
           \frac{16}{d^4}(d-1)^4, \ if \  \rho \ is \ 1-1-1-1 \ separable.
             \end{array}
            \right.
\ee
This shows that Theorem 3 in \cite{lww} is a special case of Theorem \ref{thm:7}.

{\it Example} 4. Consider the quantum state $\rho\in H_1^d\otimes\cdots\otimes H_5^d$,
\begin{eqnarray}
\rho=x|\psi\rangle\langle\psi|+\frac{1-x}{32}I_{32},
\end{eqnarray}
where $|\psi\rangle=\frac{1}{\sqrt{2}}(|00000\rangle+|11111\rangle)$ and $I_{32}$ stands for the $32\times32$ identity matrix. Because of $\|T^{(12345)}\|^2=\sum_{i_1,i_2,i_3,i_4,i_5=1}^{d^2-1}t_{i_1i_2i_3i_4i_5}^{12345}$, we see that $\|T^{(12345)}\|^2=16x^2$. By Theorem \ref{thm:7}
\be
\|T^{(12345)}\|^2\leq\left\{
             \begin{array}{lr}
               \frac{32}{d^5}(d-1)(d^2-1)^2,\ \rho\ is \ 1-4 \ or \ 1-2-2 \ separable;\\
               \frac{32}{d^5}(d^2-1)(d^3-3d+2), \  \rho \ is \ 2-3 \ separable;\\
               \frac{32}{d^5}(d-1)^2(d^3-3d+2), \ \rho \ is \ 1-1-3 \ separable;\\
               \frac{32}{d^5}(d-1)^5,\ \rho \ is \ 1-1-1-1-1 \ separable.
             \end{array}
\right.
\ee
Thus, for $\frac{3}{4}<x\leq\frac{\sqrt{3}}{2}$, $\rho$ will not be 1-4 or 1-2-2 separable. For $\frac{\sqrt{3}}{2}<x\leq1$ and $\frac{1}{2}<x\leq\frac{3}{4}$, $\rho$ will not be 2-3 separable or 1-1-3 separable respectively. For $\frac{1}{4}<x\leq\frac{1}{2}$, $\rho$ will be not 1-1-1-1-1 separable.
\section{ Conclusion  }
We have studied necessary conditions of separability for multipartite quantum states based on correlation tensors, and we have derived the upper bound for the norms of correlation tensors and the separability criterion under any partition by constructing a matrix for tripartite and four-partite quantum state. We also obtain the separability criteria under various partitions by using matrix method. Furthermore, we have studied the norm of correlation tensors for $\rho\in H_1^d\otimes\ H_2^d\otimes\cdots\otimes H_n^d$ to obtain necessary conditions of separability under any k-partition. Several examples are given under different partitions of the quantum state and our criteria are seen to be able to judge more general situations. In particular, explicit examples are given to show for tripartite and four-partite states.

\textbf {Acknowledgements}
This work is supported by the National Natural Science Foundation of China under grant Nos. 11101017, 11531004, 11726016 and 11675113,
and Simons Foundation under grant No. 523868, Key Project of Beijing Municipal Commission of Education (KZ201810028042).

\end{document}